\documentclass[conference]{IEEEtran}
\IEEEoverridecommandlockouts

\usepackage{cite}
\usepackage{amsmath,amssymb,amsfonts,bm}
\usepackage{amsthm}
\usepackage{algorithm}
\usepackage{algpseudocode}
\usepackage{graphicx}
\usepackage{textcomp}
\usepackage{xcolor}
\def\BibTeX{{\rm B\kern-.05em{\sc i\kern-.025em b}\kern-.08em
    T\kern-.1667em\lower.7ex\hbox{E}\kern-.125emX}}

\newtheorem{theorem}{Theorem}
\newtheorem{lemma}{Lemma}
\newtheorem{prop}{Claim}
\newtheorem{proposition}{Proposition}

\begin{document}

\title{A Time-Bound Signature Scheme for Blockchains}
\author{
\IEEEauthorblockN{Benjamin Marsh\IEEEauthorrefmark{1}\IEEEauthorrefmark{2}, Paolo Serafino\IEEEauthorrefmark{2}}
\IEEEauthorblockA{\IEEEauthorrefmark{1}Sei Labs\\ben@seinetwork.io}
\IEEEauthorblockA{\IEEEauthorrefmark{2}University of Portsmouth, UK\\paolo.serafino@port.ac.uk}
}
\maketitle

\begin{abstract}
We introduce a modified Schnorr signature scheme to allow for time-bound signatures for transaction fee auction bidding and smart contract purposes in a blockchain context, ensuring an honest producer can only validate a signature before a given block height. The immutable blockchain is used as a source of universal time for the signature scheme. We show the use of such a signature scheme leads to lower MEV revenue for builders.
We then apply our time-bound signatures to Ethereum's EIP-1559 and show how it can be used to mitigate the effect of MEV on predicted equilibrium strategies.
\end{abstract}


\section{Introduction}
\label{sec:intro}

Blockchains are append‑only ledgers in which a transaction becomes
\emph{final} only after its digital signature is included in a block.
Among the many signature schemes used, Schnorr stands out for its
compact proofs and linearity, and is already co‑deployed with ECDSA on
major networks.  Before a signed transaction reaches the chain,
however, it must win a place in the next block, a process usually
modelled as an auction whose game‑theoretic properties are well
studied~\cite{fdn,tfmd}. In this auction the bidders’ strategy space
is minimal: once broadcast she can only raise the fee.  This asymmetry
exposes them to \emph{maximal extractable value} (MEV): producers can front‑run, back‑run, or sandwich the transaction and
extract additional revenue, directly taxing the bidder~\cite{10.1145/3689931.3694911}.

Ethereum’s EIP‑1559 fee mechanism was expected to push tips towards zero in a steady state, yet empirical data show that
non‑negligible tips and MEV persist. The root cause is a subtle
liveness property: a bidder’s signature, once broadcast, \emph{never expires}. A rational block producer can therefore defer inclusion indefinitely, waiting for a block in which that transaction yields a higher MEV return.

Existing mitigations, commit–reveal, threshold MPC, or trusted relay
networks, either add infrastructure friction or rely on residual trust
assumptions~\cite{cryptoeprint:2019/265,cryptoeprint:2024/1533,flashprotect}.
This paper introduces \emph{time‑bound Schnorr signatures} (\mbox{TB-Sig}). By hashing an explicit expiry height into the Fiat–Shamir challenge, TB‑Sig renders a transaction invalid once the blockchain’s tip exceeds that height, giving bidders temporal control
without touching consensus rules. Because Schnorr already co‑exists
with ECDSA, TB‑Sig can be deployed as an \emph{opt‑in} upgrade that
leaves wallet key material unchanged.

Our contributions are threefold:
\begin{enumerate}
  \item Protocol: a drop‑in Schnorr variant whose signature
        tuple simply appends the expiry height.
  \item Security: proofs of existential unforgeability under
        chosen message attack in the algebraic group model, matching
        vanilla Schnorr guarantees.
  \item Economic impact: a Stackelberg analysis of EIP‑1559
        showing that TB‑Sig collapses a producer’s incentive to delay
        inclusion, and restores the bidder's incentive to bid a negligible tip.
\end{enumerate}

\subsection{Justification}
Schnorr over secp256k1 is a manageable addition to a blockchain such as Ethereum, which uses ECDSA, as has been shown by Bitcoin which supports both \cite{antonopoulos2023mastering}. Since both signature schemes use the same key material, existing wallets and hardware can adopt it with minimal effort, and none on the bidder's part. Schnorr's simple $\Sigma$-protocol proof provides a simpler security proof and small signatures, which in turn allows for trivial verification on-chain using signature size without the need for additional metadata. 
\subsection{This Work}
In this work, we propose a modification to an existing signature scheme that allows for a bidder to issue a transaction with a defined upper bound, measured in block height, for a transaction to be included in a block, allowing for revaluation of the transaction at a later time if it is not included, as well as reducing the strategy space of a non-myopic strategic block producer. The block producer must now choose between the aforementioned attack and including transactions that may not be valid in the future. The signature scheme described here relies on being verified in the same way as a standard signature for a transaction would be, and is secure under the bounds of the chain's security. Notably, the signature scheme could be optional since the signature is now represented by a tuple of 3 parts, as opposed to the original 2 parts, allowing for a dynamic decision about the signature type by the verifier. In order for the signature scheme to be used, a public blockchain is leveraged as a source of time to submit the transaction. The bidder will select a desired expiry time, as a block height, $t_e$, which must be greater than the current time, $t_e \geq t_c$ and both the value $t_e$ and a boolean value based on the evaluation of $t_c \leq t_e$ are included in the Schnorr challenge to allow for a time-based commitment meaning the signature will not validate after $t_e$. We provide a formal definition of such a signature scheme, show the inclusion of $t_e$ is secure if the underlying hash function is secure, and provide a security proof for the non-threshold version of the scheme.

We then apply our TB-Sig to the context of Ethereum's EIP-1559.
We start from the empirical observation that the predicted equilibrium behavior when EIP-1559 is in a ``steady state'' is rarely observed in real life.
To account for this, we develop a Stackelberg \cite{fudenberg1991game} model that takes into account MEV, and show that a \textit{minimum tip} is required for immediate inclusion, as otherwise a miner would be better off delaying inclusion with the aim of collecting a higher MEV from including the transaction in subsequent blocks. A Stackelberg model is chosen to allow the leader-follower separation that trivially follows from a blockchain setting where transactions are signed and added to a mempool prior to MEV searches, and inclusion in a block.
Finally, we show that introducing TB-Sig can alleviate this issue and restore the desired equilibrium strategy.
\subsection{MEV}
\subsubsection{MEV definition}
MEV, first defined by \cite{9152675}, represents the extra revenue from a transaction that a builder, or block producer, is able to extract through the manipulation of that transaction's relative order in a block through processes such as \textit{front-running}, \textit{back-running}, and \textit{sandwich attacks} \cite{10179346}. Other forms of MEV exist such as bribery but are not relevant to this work. Front-running involves the placing of an attacking transaction before the bidder's transaction in a block, be that directly before or generally earlier in the block. Back-running involves the placing of the attacking transaction behind the bidder's transaction. Sandwich attacks involve placing a transaction both before and after the bidder's transaction.
\subsubsection{MEV revenue extraction}
MEV is NP-hard even to approximate \cite{zhang2024computation}, so MEV block builders have to use heuristics and settle for suboptimal MEV extraction.
MEV bids increase monotonically over time \cite{burian2024futuremev}, as builders have longer to search for better blocks. As such, it is trivial to rationalize that the expected MEV revenue, when a time-bound signature scheme is used, is reduced as the block builder is forced to include a transaction in a sub-optimal block or risk losing the revenue from that transaction altogether. In blockchains with fast block times such as Sei, the builder may have the incentive to withhold or censor a transaction.

\subsection{Blockchain Clock}
In this work, we use the blockchain as a common source of time to provide an ordering for events, given it is accessible by all agents in the system either directly through their node or through a wallet. The block height of the tip of the chain is able to provide us with a timestamp which ticks when the next block is received. The blockchain is used in this work as a traditional Lamport Clock\cite{lamport}, providing event ordering, monotonicity, and a clock without the need to synchronize physical clocks, which the blockchain is able to do by using a shared event when a new block is added to the chain without an explicit need to synchronize time, as all nodes can refer to the event. The blockchain provides a stronger ordering bound than the partial ordering required by a Lamport Clock, as it is able to provide total ordering. Ethereum has an average block production time of 12 seconds, so in the case of Ethereum, the blockchain clock ticks once every 12 seconds. The use of the blockchain as a clock dates back to Bitcoin\cite{btc} where the \textit{nLockTime} OP code is used to prevent the inclusion of a transaction until a defined block height. The blockchain clock is secured by the underlying economic security\cite{abadi2018blockchain} of the chain and can be assumed to be secure if the chain itself is secure against attacks via the rewards and punishments of the consensus protocol. As a result, the blockchain clock can be trusted to provide a reliable and tamper-resistant source of relative temporal information. The use of a universal clock may lead to stale views of the chain in some circumstances meaning the transaction is submitted having an already expired signature, such behavior is acceptable allowing the user to resubmit with a true valuation in a state of complete information. A subtlety arises in the presence of short reorgs. Suppose a transaction
with expiry $t_e = n$ is included in a block at slot $n$, but this block
is later orphaned in slot $n{+}1$. By the time the competing block at slot
$n{+}1$ becomes canonical, the transaction has expired and can no longer
be included.
If inclusion remains desirable, the user then needs to reissue the transaction with a fresh
signature. This behavior is consistent with
our design intent: the expiry bound is conservative, and resubmission simply
re-elicits the user’s true valuation under more up-to-date information.
Users who wish to guard against short reorgs can set $t_e$ a few blocks ahead
($t_e \geq t_c + k$ for small $k$), thereby trading tighter expiry against
resilience to transient forks.
\subsection{Related Work}
\subsubsection{Cryptography}
Recent advancements in Schnorr signatures have primarily focused on threshold signatures, such as FROST \cite{frost}, which have introduced novel modifications and cryptographic techniques widely accepted in the blockchain industry. These advancements have paved the way for new cryptographic methods and implementations in secure transactions. Zero-knowledge cryptography and non-interactive proofs, as explored in \cite{shortlived}, propose short-lived proofs and signatures that could be used to construct a protocol similar to the one we propose, relying on a verifiable delay function. Their approach offers a more general framework, whereas our method is simpler to implement and verify within a context where a blockchain clock is viable. Time-bound cryptography has also been explored extensively, starting with Rivest's work \cite{timelock} on time-lock puzzles, which introduced timed-release cryptography. This work aimed at ensuring that a signature could not be verified \textit{until} a specific time had passed, which contrasts with our proposal that a signature is only valid \textit{up to} a specific time. More recent studies, such as \cite{tide}, have continued to develop these concepts, exploring various mechanisms for implementing time-bound signatures in different cryptographic contexts. Our proposed modification to the Schnorr signature scheme leverages these advancements to introduce a time-bound element that allows for a transaction to be re-evaluated if not included within a specified time frame, addressing strategic vulnerabilities in blockchain transaction inclusion.
\subsubsection{MEV Countermeasures}
Numerous MEV countermeasures have been proposed including the use of cryptographic protocols such as commit-reveal schemes \cite{cryptoeprint:2019/265} and multi-party computation \cite{cryptoeprint:2024/1533} yet such protocols have limited usage outside of academia due to the overhead of such protocols. Fair ordering based on arrival time \cite{cryptoeprint:2020/269}, random ordering \cite{kavousi2023blindperm}, and batch auctions \cite{budish2015high} have also been proposed in the literature with nothing but niche industry acceptance. Centralized side-protocol approaches such as Flashbots protect \cite{flashprotect} also exist where a transaction is routed through a private RPC end-point and mempool, but they introduce a risk of failure or censorship due to their centralized nature. Protect also strictly enforces limits and drops zero-tip transactions, complicating wallet integration.
\subsection{Organization of the Paper}
The remainder of the paper is organized as follows. In Section \ref{sec:background} we give some background on Schnorr signature schemes. In Section \ref{sec:TB‑Sig} we introduce our TB‑Sig, whereas in Section \ref{sec:security} we prove its security guarantees. In Section \ref{sec:stack_model} we propose a Stackelberg model for Ethereum's EIP-1559 that takes into account MEV, analyze its equilibrium strategies, 
and show how TB‑Sig can counteract MEV.
Finally, in Section \ref{sec:conclusions} we draw some conclusions.

\section{Background}\label{sec:background}
\subsection{Notation}
In line with cryptography literature, we write  \(x \overset{{\scriptscriptstyle \mathsf{\$}}}{\leftarrow} S\) to denote that \(x\) is sampled uniformly at random from the finite set \(S\). By $\mathbb{Z}_q$ we denote the finite field of integers modulo a prime $q$, whereas \(\mathbb{G}\) denotes a cyclic group of prime order \(q\) (with generator \(g\)), in which the discrete logarithm problem is assumed hard. 
A \textit{random oracle} (\(\mathsf{RO}\)) is an ideal hash oracle \(H(\cdot)\) that returns fresh uniform outputs on the first query and is consistent thereafter.
We use the notation $H(a,b)$ to denote the hash of the concatenation of $a$ and $b$. 
We use the term \textit{producer} for an agent who creates a block, sometimes called a miner, validator, builder, or proposer in other works.
We use \textit{bidder} to denote the transaction issuing agent bidding in the inclusion auction.
We will assess the security of our proposed TB-Sig against an adversary $\mathcal{A}$.
Unless otherwise stated, $\mathcal{A}$ is assumed to be a \textit{probabilistic polynomial time Turing machine} (PPT) with random tape $\omega$.
\subsection{The Schnorr Signature Scheme}
A Schnorr signature\cite{Schnorr} is generated for a message $m$, under a secret key $s \in \mathbb{Z}_q$ and a public key $Y = g^s \in \mathbb{G}$ by the following protocol, assuming the existence of a cryptographically secure hash function $H$:
\begin{enumerate}
    \item Sample a random nonce, $k \overset{{\scriptscriptstyle \mathsf{\$}}}{\leftarrow} \mathbb{Z}_q$ and compute the commitment $R = g^k \in \mathbb{G}$
    \item Compute the challenge $c = H(R, Y, m)$
    \item Using the secret key, $s$, compute the response $z = k + s \cdot c \in \mathbb{Z}_q$
    \item The signature of $m$ is thus defined as $\sigma = (R, z)$ 
\end{enumerate}
To verify the signature $\sigma$ of a message $m$ with a public key $Y$ we use the following protocol:
\begin{enumerate}
    \item Given $\sigma = (R, z)$, compute $c = H(R, Y, m)$
    \item Compute $R' = g^z \cdot Y^{-c}$
    \item Output 1 if $R = R'$ to indicate success, else output 0.
\end{enumerate}
Schnorr signatures are simply the standard $\Sigma$-protocol proof of knowledge of the discrete logarithm of \( Y \), made non-interactive (and bound to the message \( m \)) with the Fiat-Shamir transform \cite{FS87}, and can be reduced to the discrete logarithm problem in the random oracle model according to \cite{PS96}.

\section{Time-bound signatures}\label{sec:TB‑Sig}
We now present a time-bound signature scheme that allows for time-bound agent interactions with blockchain-centric protocols and auctions, such as transaction fee mechanisms. For the time-bound signature scheme to work, a blockchain clock must be available to all participants, whether they are signing or verifying a signature. As such, our time-bound signature scheme relies on the block height of a public blockchain. In order to account for time, the signature scheme includes the output of a time-checking function in the challenge, relying on the avalanche effect to ensure the signature will not be verified when the signature is required to no longer be valid. With the exception of two new fields in the Schnorr commitment and the requirement for a blockchain clock, the signature scheme is otherwise unchanged in its usage, and the overhead from the modifications is negligibly cheap both computationally and from a storage perspective.
\subsection{Time check function}\label{tchk}
The time check function $f_t$ takes two inputs: the current block height $t_c$ and the desired end block height $t_e$. During the signing process, it is assumed that the current block height is strictly less than the desired end block height, i.e., $t_c < t_e$. This condition should be enforced by end-user software. The end validity height is included as part of the new signature $\sigma = (R, z, t_e)$ to allow a block producer to ensure a transaction is valid for the desired block of inclusion, or for mempools to handle expired transactions. When a signature is verified, the current block height can be taken from the height of the block being produced, or that is being verified.
The time check function is defined as follows:
\begin{equation}
f_t(t_c, t_e) = \begin{cases}
1 & \text{if } t_c \leq t_e \\
0 & \text{otherwise}
\end{cases} \label{eq:tfunc}
\end{equation}
\subsection{Signing}
The signing process for the time-bound signature scheme closely mirrors the traditional Schnorr signature scheme, with the addition of a pair of extra terms to the challenge. The inclusion of the $t_e$ term ensures a cryptographic commitment to the end value exists in such a way that a block producer cannot trivially spoof a different end time. The time check function term allows for the verification of the current time against the desired time. Assuming $H$ is the cryptographically secure hash function SHA256, the signing process is as follows:
\begin{enumerate}
    \item Sample a random nonce, $k \overset{{\scriptscriptstyle \mathsf{\$}}}{\leftarrow} \mathbb{Z}_q$ and compute the commitment $R = g^k \in \mathbb{G}$
    \item Compute the challenge $c = H(R, Y, m, t_e, f_t(t_c, t_e))$
    \item Using the secret key, $s$, compute the response $z = k + s \cdot c \in \mathbb{Z}_q$
    \item The signature of $m$ is thus defined as $\sigma = (R, z, t_e)$ 
\end{enumerate}
\subsection{Verification}
Verification also closely mirrors the original Schnorr scheme. 
To verify the signature $\sigma=(R, z, t_e)$ of a message $m$ with a public key $Y$, use the following protocol:
\begin{enumerate}
    \item Compute ${c = H(R, Y, m, t_e, f_t(t_c, t_e))}$
    \item Compute $R' = g^z \cdot Y^{-c}$
    \item Output 1 if $R = R'$ to indicate success, else output 0.
\end{enumerate}
The verification process will fail at the $R = R'$ check if: (\textit{i}) the current block height is higher than the end time, i.e., $t_c > t_e$; or (\textit{ii}) if the producer checks a different value for $t_e$.
Introducing TB-Sig requires the following minimal modifications to the consensus protocol.
\subsubsection{Block producer}
When a block producer is producing a block, they must now ensure the transaction is valid for inclusion; otherwise, the block will be rejected by the network consensus rules. This can be done trivially cheaply by checking that $f_t(t_c, t_e) = 1$. 

\subsubsection{Block validation}
When a block is received by the network either during the normal block production process or during syncing, each block will be validated. To validate the signatures for a given block, any time-bound signature should be compared to the height of the block in which it is included.

\section{Security}\label{sec:security}
\subsection{Inclusion of the time function}
\begin{prop} 
\textit{The inclusion of $f_t$ does not weaken the challenge.}
\end{prop}
\begin{proof}[Sketch]
The time check function as defined in \eqref{eq:tfunc} is a deterministic and easily predictable bit, its inclusion
in the Fiat-Shamir challenge does not reduce security. Formally, the
challenge is
$
c = H(R,Y,m,t_e,f_t(t_c,t_e)).
$

Security relies on the standard properties of the hash function $H$, as defined in \cite{boneh2020graduate}:
\begin{itemize}
    \item \textbf{Pre-image resistance:} given $h$, finding
    $(R,Y,m,t_e,f_t)$ such that $H(R,Y,m,t_e,f_t)=h$ is infeasible.
    \item \textbf{Second pre-image resistance:} given
    $(R,Y,m,t_e,f_t)$, finding
    $(R',Y',m',t'_e,f'_t)\neq(R,Y,m,t_e,f_t)$ with the same hash output
    is infeasible.
\end{itemize}
Since $f_t$ contributes only a single bit, the adversary gains no advantage:
the entropy of the challenge is still dominated by $(R,Y,m,t_e)$. The value
of $f_t$ merely hard-codes into the signature whether the expiry condition
was satisfied at verification. Thus the scheme’s security reduces to the
collision resistance of $H$, and Schnorr’s EUF–CMA guarantees are preserved.
\renewcommand{\qedsymbol}{}
\qedhere
\end{proof}
\subsection{Commitment to the desired final validity height}
\begin{prop}
\textit{The inclusion of $t_e$ in the Schnorr challenge is required as a commitment to the required signature time-bound.}
\end{prop}
\begin{proof}[Sketch]
Suppose $t_e$ were excluded from the Fiat--Shamir challenge. Then a block
producer given a valid signature $\sigma=(R,z,t_e)$ could simply replace
$t_e$ by some $t'_e > t_c$, since the pair $(R,z)$ would still verify
against $Y$. This trivial malleability would allow an adversary to extend
the validity window indefinitely and strategically delay inclusion.

By hashing the end block height $t_e$ into the challenge
\(
c = \nolinebreak H(R,Y,m,t_e,f_t(t_c,t_e)),
\)
any change of $t_e$ necessarily alters the challenge input. To produce a
forged transcript $\sigma'=(R,z,t'_e)$, an adversary must find a second
pre-image satisfying
\(
H(R,Y,m,t'_e,1) = H(R,Y,m,t_e,1).
\)
Under the collision-resistance of $H$, this occurs with probability at most
$2^{-n}$ for security parameter $n$. Thus expiry is cryptographically bound
to the signature and cannot be malleated, ensuring that the producer cannot
extend a transaction’s validity beyond the signer’s intent.
\renewcommand{\qedsymbol}{}
\qedhere
\end{proof}
\subsection{Blockchain clock}
\begin{prop}
    A decentralized blockchain is a valid blockchain clock for our signature scheme.
\end{prop}
The signature scheme requires a universally available blockchain clock provided by a public blockchain. Attacking the signature scheme would require \textit{pushing the block height on} to invalidate the signature prematurely, meaning the attacker has the ability to break the consensus mechanism and produce blocks at will. 
This would require a fundamental attack against hash-based PoW mechanisms or PoS mechanisms such as VRFs. 
Thus, it is safe to conjecture the signature scheme is secure as long as the blockchain is. 
A signature is also attackable by \textit{reducing the block height}; this would require a reorganization of the chain and the creation of a new chain that is selected as the canonical chain; by design, this is unlikely through the chains' respective chain selection rules or finality rules.
The clock also ticks at the pace of the auction once every produced block, making it a viable clock for the setting.
\subsection{Validation by Chain Consensus Rules}
\begin{prop}
The time-bound signature scheme is secure under a blockchain's consensus security bounds.
\end{prop}
An adversary $\mathcal{A}$ may attempt to submit a block for validation to the network that includes an invalid transaction where $t_c > t_e$. The acceptance or rejection of that block depends on the blockchain network's consensus rules.
The security of the time-bound signature scheme relies on the underlying blockchain's consensus mechanism and inherits its security. We consider two primary types of consensus protocols: chain-based consensus (such as Bitcoin) and Byzantine Fault Tolerant (BFT) consensus protocols.

\subsubsection{Chain-based Consensus}
    In a chain-based consensus protocol, security is typically guaranteed under the assumption that an adversary controls less than 50\% of the total hashing power/relative stake. The signature scheme is secure under these conditions because:
    \begin{itemize}
        \item \textbf{Block Validation:} If an adversary attempts to include an invalid transaction (with $t_c > t_e$) in a block, the majority of honest block producers will reject the block based on consensus rules.
        \item \textbf{Attack Resistance:} The scheme is secure against $50\%+\epsilon$ attacks, meaning any attack that requires more than $50\%$ of the network's hashing power is considered infeasible under typical security assumptions. Formally, let $Pr(\mathcal{A})$ denote the probability that an adversary $\mathcal{A}$ can successfully include an invalid transaction. For security, we require \(Pr(\mathcal{A}) \leq \epsilon\) where $\epsilon$ is a negligible probability.
    \end{itemize}
\subsubsection{BFT-style Consensus Protocols}
    In BFT consensus protocols, security is typically guaranteed as long as no more than 33\% of the nodes are malicious. The signature scheme conforms to these security bounds because:
    \begin{itemize}
        \item \textbf{Block Validation:} In BFT protocols, an invalid transaction will not be accepted unless at least 33\% of the producers are colluding. Let $N$ be the total number of producers and $f$ be the number of malicious producers. For security, we require \(f < \frac{N}{3}\)
        \item \textbf{Attack Resistance:} The scheme is secure against attacks where up to 33\% of producers are malicious. The probability $Pr(\mathcal{A})$ that an adversary $\mathcal{A}$ can successfully include an invalid transaction is negligible if $f < \frac{N}{3}$.
    \end{itemize}
Allowing strategic behavior by a block producer or a cartel of block producers requires collusion at a level capable of more significant attacks than merely including an out-of-date, time-bound transaction. If a cartel of block producers includes an invalid transaction in a block and this block is validated, it could lead to a chain fork. Honest block producers joining the network will likely reject the invalid block, maintaining the integrity of the blockchain.
\subsubsection{Forking}
Assume that $\alpha$ is the fraction of honest block producers and $\beta$ is the fraction of adversarial block producers, with $\alpha + \beta = 1$. The probability of an honest block being added to the longest chain is \(Pr(h) = \frac{\alpha}{\alpha + \beta}\).
If $\alpha > 0.5$, the longest chain will eventually consist of honest blocks, as the probability of honest blocks outpacing adversarial blocks increases over time, reducing the non-myopic benefit of attempting to include an invalid transaction. 
Thus, the signature scheme's security is inherently tied to the security model of the underlying blockchain. Under the assumption of an honest majority ($50\%+\epsilon$ for chain-based consensus and $2f+1$ for BFT protocols), the time-bound signature scheme remains secure. The possibility of including invalid transactions through collusion requires a level of attack that would compromise the entire blockchain, making the scheme as secure as the consensus mechanism it relies on.
\subsection{EUF-CMA}
\subsubsection{EUF–CMA game for the time-bound Schnorr scheme}\label{sec:EUF-CMA-game}
Let ${\sf TB\text{-}Sig}$ denote the time-bound Schnorr scheme.  
$\mathcal{A}$ has adaptive access to:
        \begin{itemize}
           \item a hash oracle ${\sf RO}(\cdot)$ implementing the
                 random oracle $H$;
           \item a signing oracle
                 ${\sf Sign}(m,t_e)$ that returns a valid signature
                 $\sigma=(R,z,t_e)$ on the pair $(m,t_e)$ provided
                 $(m,t_e)$ has not been queried before. %
                 There is no restriction on the relation between
                 the real-world block height $t_c$ and the input $t_e$
                 inside the experiment; the flag $f_t(t_c,t_e)$ is fixed
                 to~$1$ inside ${\sf Sign}$.
            \item ${\sf Verify}_{Y}(m,\sigma,t_c)$ which verifies that $\sigma$ is a valid signature for $m$ at block $t_c$.
        \end{itemize}
The EUF–CMA experiment $\mathsf{EUF}^{\mathcal{A}}_{TB\text{-}Sig}(n)$ between a
probabilistic adversary $\mathcal{A}$ and a challenger $\mathcal{C}$ proceeds
as follows:

\begin{enumerate}
  \item \textit{Key generation.} 
        $\mathcal{C}$ runs ${\sf KeyGen}(1^n)$, obtaining $(s,Y)$ with
        $Y=g^{s}\!\in \mathbb{G}$.  It gives $Y$ to $\mathcal{A}$.
  \item \textit{Query.} $\mathcal{A}$ may repeatedly ask for signatures on chosen messages $(m_1, \dots, m_q)$ of its choosing, and receives the valid signatures $(\sigma_1, \dots, \sigma_q)$ in response.
  These queries and responses can be made adaptively and interactively.
  \item \textit{Forgery.} 
        Eventually $\mathcal{A}$ outputs $(m^\star,t_e^\star,\sigma^\star)$.
        Let $t_c^\star$ be the block-height parameter supplied by the
        experiment to ${\sf Verify}$.  
        $\mathcal{A}$ wins if
        $ (m^\star,t_e^\star)$
            was never given to \\${\sf Sign}$ 
            in the Query phase, and $
            {\sf Verify}_{Y}(m^\star,\sigma^\star,t_c^\star)=1$.
\end{enumerate}
The advantage $\mathrm{Adv}_{TB\text{-}Sig}^{\mathrm{EUF}}$ of adversary $\mathcal{A}$ is
$\mathrm{Adv}_{TB\text{-}Sig}^{\mathrm{EUF}}(\mathcal{A})
 =\Pr[\mathsf{EUF}^{\mathcal{A}}_{TB\text{-}Sig}(n)=1]$.
 
\subsubsection{Signing-oracle simulation}

Let \(\mathcal{S}\) be the reduction that is given the challenge public key \(Y=g^{s}\) with unknown \(s\) and interacts
with a forging adversary \(\mathcal{A}\).

\begin{algorithm}[H]
\caption{$\mathcal{S}.{\sf Sign}(m,t_e)$ without knowledge of $s$\label{alg:signsim}}
\begin{algorithmic}[1]
  \State Pick $c \xleftarrow{\$}\! \mathbb{Z}_q$ and
         $z \xleftarrow{\$}\! \mathbb{Z}_q$
  \State Compute $R \gets g^{z}\,Y^{-c}$
  \State Set $H(R,Y,m,t_e,1) := c$ in the random oracle.
  \State Return $\sigma=(R,z,t_e)$
\end{algorithmic}
\end{algorithm}

We ensure $(R,c,z)$ satisfies  
$g^{z}=R\,Y^{c}$; programming the oracle guarantees
$\sigma$ will verify.  The distribution of $\sigma$ is identical to that
of a real signature because $(R,c)$ is uniform and independent.
\(\mathcal{S}\) therefore perfectly simulates ${\sf Sign}$ for
\(\mathcal{A}\).
\subsubsection{Algebraic Group Model}
The Algebraic Group Model\cite{agm} assumes that all adversaries are \textit{algebraic}. Formally, an algorithm \( \mathcal{A} \) is algebraic if whenever it outputs a group element \( A \in \mathbb{G} \), it also outputs the vector \( \mathbf{a} = (a_1, \ldots, a_n) \) such that \( A = \prod_{i=1}^n h_i^{a_i} \) for some known group elements \( h_1, \ldots, h_n \in \mathbb{G} \). That is, the adversary not only outputs the group element but also a linear combination of all previously seen group elements that generate it.
\begin{lemma}[Time-bound Forking Lemma]
\label{lem:tbFork}
Let $\mathcal{A}$ be any PPT adversary that makes at most
$q_H$ random-oracle queries and $q_S$ signing-oracle queries in the EUF–CMA experiment
of Section~\ref{sec:EUF-CMA-game}, and that outputs a valid
forgery with a probability $\epsilon$. 
Running $\mathcal{A}$ in the rewinding algorithm
${\sf Fork}^{\mathcal{A}}$ defined exactly as in~\cite{PS96}
but with the hash input
\(
   (R,Y,m,t_e,f_t)
\)
yields, with probability at least
\(
   {\epsilon^2 \Big/ \bigl(q_H+q_S+1\bigr)}
\),
two accepting transcripts
\(
   (R,c_1,z_1,t_e) \neq (R,c_2,z_2,t_e)
\)
such that $c_1\neq c_2$ and
\(
   g^{z_i}=R\,Y^{c_i}\ (i=1,2).
\)
\end{lemma}
\begin{proof}
Let \(q = q_H + q_S + 1\) as in~\cite{cryptoeprint:2012/029}. We define the forking reduction ${\sf Fork}^{\mathcal{A}}$ as follows. Run $\mathcal{A}$ once on fresh randomness, answer all oracle queries honestly, and record both the sequence of random‐oracle replies $(h_1,\dots,h_q)$ and $\mathcal{A}$’s random tape.  If no valid forgery appears, the reduction aborts.  Otherwise, it picks a uniform index $j\in\{1,\dots,q\}$, rewinds $\mathcal{A}$ to just before its $j$-th random‐oracle call, and returns a fresh independent $n$-bit reply $h'_j$ while replaying all other replies and the original tape.  If the second run also outputs a forgery, the two transcripts are returned; otherwise the reduction aborts. 
Now, let $\varepsilon$ be the probability that the first execution forges, and let $j^\star$ be the (unique) position where the challenge $c^\star=H(R,Y,m,t_e,f_t)$ is first queried.  Conditioned on a successful first run, the probability of choosing $j=j^\star$ is $1/q$.  In that case, the fresh reply $h'_j$ makes the second challenge $c_2$ independent of $c_1$, so $\Pr[c_2\neq c_1]=1-2^{-n}$.  Thus:
\begin{align*}\small
  \Pr[\text{2 distinct transcripts}]
  &\ge \varepsilon \times \frac{1}{q} \times (\varepsilon (1 - 2^{-n}))
  &\ge \frac{\varepsilon^2}{q}\,
\end{align*}
absorbing the negligible $2^{-n}$ term for any practical $n\ge128$.  The total work is at most twice that of $\mathcal{A}$ plus $q$ extra oracle calls.  Since queries before $j$ (and all signing replies) are replayed identically, each transcript satisfies $g^{z_i}=R\,Y^{c_i}$ with $c_1\neq c_2$, completing the fork.
\end{proof}
\subsubsection{AGM security model}
\begin{theorem}
Assuming the discrete logarithm problem in \( \mathbb{G} \) is hard, if there exists an algebraic adversary \( \mathcal{A} \) that can produce a valid forgery for the time-bound Schnorr signature scheme with non-negligible probability, then there exists an algorithm that can solve the discrete logarithm problem in \( \mathbb{G} \) with non-negligible probability. That is, if \( \mathcal{A} \) can output a forgery \( (m^*, \sigma^* = (R^*, z^*, t_e^*)) \) for a new message \( m^* \) not queried to the signing oracle, then there exists an algorithm \( \mathcal{S} \) that, given \( Y = g^s \), can compute \( s \) with non-negligible probability.
\end{theorem}
\begin{proof}
Let \( \mathbb{G} \) be a cyclic group of prime order \( q \) with generator \( g \). The private key \( s \) is chosen uniformly at random from \( \mathbb{Z}_q \), and the public key is \( Y = g^s \). We assume a cryptographically secure hash function \( H \) modeled as a random oracle.\\
The simulator \( \mathcal{S} \) receives a challenge \( Y \) (where \( Y = g^s \) and \( s \) is unknown). \( \mathcal{S} \) interacts with the adversary \( \mathcal{A} \) as follows:\\
Setup:
\begin{enumerate}
    \item The simulator \( \mathcal{S} \), \( Y = g^s \), interacts with \( \mathcal{A} \) as described.
    \item \( \mathcal{S} \) simulates the signing oracle for \( \mathcal{A} \) and maintains a list of queried message-time pairs.
\end{enumerate}
Forgery:
\begin{enumerate}
    \item Suppose that the adversary \( \mathcal{A} \) outputs a valid forgery \( {(m^*, \sigma^* = (R^*, z^*, t_e^*))} \) that was not queried. By Lemma~\ref{lem:tbFork}, with high probability we can obtain two valid signatures \( (R^*, z_1, t_e^*) \) and \( (R^*, z_2, t_e^*) \) with different challenges \( c_1 \) and \( c_2 \).
    \item We can then write \( g^{z_1} = R^* \cdot Y^{c_1} \) and \( g^{z_2} = R^* \cdot Y^{c_2} \), from which we derive
    $g^{z_1 - z_2} = Y^{c_2 - c_1}$. From the latter we finally obtain $s = \frac{z_1 - z_2}{c_2 - c_1}$.
    

\end{enumerate}
By extracting \( s \), we solve the discrete logarithm problem, thus proving the security of TB-Sig under the discrete logarithm assumption in the AGM.
\end{proof}

\section{Stackelberg Model with MEV Extraction}\label{sec:stack_model}


In this Section, we study a simple Stackelberg game \cite{fudenberg1991game} that models the strategic interaction between bidders (the leaders) and a block producer (the follower) in the EIP-1559 Transaction Fee Mechanism.

In EIP-1559, each new block is associated with a base fee $b_f$, which is determined based on on-chain history, and has a constant target capacity $s^*$. The maximum quantity of gas allowed in a block is $2s^*$.
Bidder $i$ has a valuation $\bar{v}_i$ for her transaction being included in a block, and bids two values: (\textit{i}) a \textit{fee cap} $b^c_i$, which represents the maximum amount of currency per unit of gas that she is willing to pay; and (\textit{ii}) a \textit{tip} $\tau_i$ which represents the maximum tip per unit of gas.
A transaction can be included in a block only if its fee cap is at least the block's base fee, i.e., $b^c_i\geq b_f$, and pays $\min\{b_f+\tau_i,b^c_i\}$.
A producer is selected to produce a block with a probability $\rho$ that is proportional to her stake. Upon producing a new block, the producer gets: (\textit{i}) a fixed \textit{block reward}, which is independent of the transactions included in the block (hence omitted from our model as not relevant to our discussion); and (\textit{ii}) $\min\{\tau_i, b^c_i-b_f\}$ for each transaction $i$ included in the block.
We assume that the system is in a \textit{steady state} as defined in \cite{AzouviGHH23}. Essentially, this means that the \textit{demand curve} for a unit of gas in the next block \textit{does not change over time}. 
The expected dynamics of EIP-1559 in a steady state are the following: (\textit{i}) all blocks produced are of target size $s^*$; (\textit{ii}) the base fee is constant. The previous two conditions imply that the \textit{base fee is not excessively low}, as defined in \cite{tfmd}. 
The \textit{predicted equilibrium bidding} in this state is the following: (\textit{i}) the fee cap is equal to the bidder's valuation, i.e., $b^c_i = \bar{v}_i$; (\textit{ii}) the tip is a negligible amount $\tau_i >0$.
Roughly speaking, EIP-1559 effectively acts as an \textit{unlimited-supply posted-price} mechanism that \textit{truthfully elicits the fee cap} from each bidder.


The producer's strategy space consists of two strategies: (\textit{i}) \textit{immediate inclusion}, where the producer includes a transaction immediately, collecting the tip; and (\textit{ii}) \textit{delayed inclusion}, where the producer decides to include the transaction in a later block, thus collecting the tip, plus any MEV that has been able to be extracted in the meantime.
%
To model MEV extraction, we use a function $\Delta_i : \mathbb{N} \rightarrow \mathbb{R}$, where $\Delta_i(t)$ represents the MEV that can be extracted from transaction $i$ after $t$ blocks.
We assume that: (\textit{i}) $\Delta_i(t)=0$ for $t\leq t_i$, where $t_i$ is the block height at which the transaction has been issued; (\textit{ii}) $\Delta_i(t)\rightarrow \delta_i$ as $t \rightarrow \infty$, i.e. $\Delta_i$ has a horizontal asymptote which represents the maximum MEV that can be extracted from the transaction; and (\textit{iii}) $\Delta_i$ is monotone non-decreasing
(as argued in \cite{schwarz_schilling2023timing, burian2024futuremev}).
We assume that the producer has access to an \textit{estimate of the MEV} that can be extracted by transaction $i$ if the inclusion of $i$ is delayed to a later block. With a little abuse of notation we denote this estimate as $\Delta_i$.
%
For a generic future block, we denote as $I$ the random variable representing the \textit{minimum value for inclusion}, which is the minimum value to the block producer (consisting of tip plus MEV) at which the block producer will include the transaction in the block. We assume that $I$ is distributed according to the cumulative distribution function $F_I$, which is continuous and defined over a bounded support, and its probability density function $f_I$ is monotone non-decreasing.
We represent the utility $u_i(\tau_i,\textbf{x})$ of bidder $i$ as a function of $i$'s bid and the producer's strategy $\textbf{x}=(x_1,\ldots,x_{s^*})$, where $x_i = 1$ if transaction $i$ is included immediately, and $x_i=0$ otherwise.
The utility function for bidder $i$ is the following:
\begin{equation}
u_i(\tau_i,\textbf{x})= x_i(v_i - \tau_i)+(1-x_i)F_I(\tau_i+\Delta_i)\,(v_i - \tau_i)
\label{eq:bidder_utility}
\end{equation}
where $v_i = \bar{v}_i-b_f$.
This means that bidder $i$'s utility is $v_i-\tau_i$ if its transaction is included immediately, and is $F_I(\tau_i+\Delta_i)\,(v_i - \tau_i)$ in expectation otherwise.
Notice that here the utility of $i$ depends solely on the tip $\tau_i$ that she bids to pay.
This is consistent with the fact that the base fee, by assumption, is constant over the period of interest of our analysis.
Also notice that we assume bidders to be \textit{patient}, i.e., their valuation does not decrease as a consequence of delayed inclusion.

The producer's utility $u_M(\bm{\tau},\textbf{x})$ is instead a function of the bid vector $\bm{\tau}=(\tau_1,\ldots,\tau_{s^*})$ and of the producer's strategy $\textbf{x}$, and is given by the following:
\begin{equation}\small
u_M(\bm{\tau},\textbf{x})=\sum_{i=1}^{s^*} \left(x_i\cdot \tau_i+ (1-x_i)F_I(\Delta_i+\tau_i)(\Delta_i+\tau_i)\rho\right)    
\label{eq:miner_utility}
\end{equation}
In particular, for each transaction $i$ that is immediately included, $M$ gets the tip $\tau_i$, whereas for each transaction whose inclusion is delayed, $M$ gets, in expectation, $F_I(\Delta_i+\tau_i)(\Delta_i+\tau_i)\rho$.
We will also consider the utility of $M$ when a \textit{transaction is obtained from a private mempool}, i.e.,
the transaction is known only to $M$.
The following lemma proves that, in this setting, there is a lower bound to how low a tip can be if a transaction is to be included immediately.
\begin{lemma}\label{thm:min_bid}
(\textit{i}) If EIP-1559 is in a steady state, it is a best response for the producer to immediately include a transaction $i$ only if $\tau_i> \Delta_i\frac{F_I(\Delta_i+\tau_i)\cdot\rho}{1-F_I(\Delta_i+\tau_i)\cdot\rho}$. 
(\textit{ii}) If $i$ was obtained from a private mempool, then $\tau_i> \Delta_i\frac{F_I(\Delta_i+\tau_i)\cdot}{1-F_I(\Delta_i+\tau_i)\cdot}$
\end{lemma}
\begin{proof}
We only need to consider the $s^*$ transactions that are eligible for inclusion in the current block.

(\textit{i}) From \eqref{eq:miner_utility}, in order to maximize her utility, the producer will include $i$ immediately if $\tau_i >F_I(\Delta_i+\tau_i)(\Delta_i+\tau_i)\rho$, and delay its inclusion otherwise. Solving the above for $\tau_i$, we get that $\tau_i> \Delta_i\frac{F_I(\Delta_i+\tau_i)\cdot \rho}{1-F_I(\Delta_i+\tau_i)\cdot \rho}$.

(\textit{ii}) We observe that in this case the utility of transaction $i$ to the producer is $x_i\cdot \tau_i+ (1-x_i)F_I(\Delta_i+\tau_i)(\Delta_i+\tau_i)$, since $i$ is not known to other producers, and $M$ will eventually produce a block where they can include $i$. A similar argument to part (\textit{i}) can then be applied.
\end{proof}
Notice that, depending on $\Delta_i$ and $F_I$, the minimum tip for immediate inclusion can be quite large.
The following lemma proves that, irrespective of $F_I$ and assuming that $\tau_i \leq \Delta_i$, there is always a value of $\Delta_i$ for which delayed inclusion is a dominant strategy for the producer.
\begin{proposition} \textbf{(Public mempool)}
For any $F_I$, any tip $\tau_i$, and any relative stake $\rho$, there exists a value of MEV $\Delta_i$ such that delayed inclusion is a best response for the producer. In particular, if $\Delta_i >\frac{\tau_i}{F_I(\tau_i)\cdot\rho}$, then the producer's best response is delayed inclusion.
\end{proposition}
\begin{proof}
Assume that $\Delta_i > \frac{\tau_i}{F_I(\tau_i)\cdot\rho}$. Then we have:
$$\small
\Delta_i > \frac{\tau_i}{F_I(\tau_i)\cdot\rho}\geq\frac{\tau_i}{F_I(\Delta_i+\tau_i)\cdot\rho} \geq \tau_i\frac{1-F_I(\Delta_i+\tau_i)\cdot \rho}{F_I(\Delta_i+\tau_i)\cdot\rho}
$$
where the second inequality follows from the monotonicity of $F_I$, and the last inequality from $F_I(\Delta_i+\tau_i)\cdot \rho \geq 0$. Rearranging the above chain of inequalities, we get the following:

$$\small
\Delta_i \frac{F_I(\Delta_i+\tau_i)\cdot \rho}{1-F_I(\Delta_i+\tau_i)\cdot\rho} > \tau_i
$$

The claim now follows from Lemma \ref{thm:min_bid}.
\end{proof}

%
%
For the case when transaction $i$ is received from a private mempool, the following proposition proves that delayed inclusion is a best response for the producer for values of MEV that are at least the 50th percentile of $F_I$.

\begin{proposition} \textbf{(Private mempool)}
For any $F_I$ and any tip $\tau_i$, there exists a value of MEV $\Delta_i$ such that delayed inclusion is a best response for the producer. In particular, if $\Delta_i> \tau_i$ and $F_I(\Delta_i)>\frac{1}{2}$, then delayed inclusion is a best response for the producer.
\end{proposition}
\begin{proof}
\newcommand{\deltaT}{\Tilde{\Delta_i}}
Let us consider a value of $\Delta_i$ such that: $\Delta_i > \tau_i$ and $F_I(\Delta_i) > \frac{1}{2}$.
From the latter, it is easy to see that $\frac{F_I(\Delta_i+\tau_i)}{1-F_I(\Delta_i+\tau_i)}>1$, which, combined with the former, yields:

$$
\tau_i < \Delta_i\frac{F_I(\Delta_i+\tau_i)}{1-F_I(\Delta_i+\tau_i)}
$$
From Lemma \ref{thm:min_bid} it follows that the producer's best response is delayed inclusion.
\end{proof}

We will now argue that introducing our time-bound signature scheme can help reduce the effect of MEV on tips.
The strategy space of a bidder is now expanded to include, besides fee cap $b_i^c$ and tip $\tau_i$, also a time $t_e$ after which the transaction can no longer be included in a block.
This, of course, has an effect on the utility functions of both the producer and the bidders, and, as a consequence, changes the equilibria of the ensuing game.
The following proposition proves that if EIP-1559 is in a steady state, then the equilibrium strategy will be the following: (\textit{i})
bidder $i$ bids a fee cap $b^c_i=\bar{v}_i$, a tip $\tau_i=0$ and a time limit $t_e=t_c$, where $t_c$ is the current block height; the producers will include immediately all transactions $i$ such that $b_i^c\geq b_f$.

\begin{proposition}
If EIP-1559 is in a steady state, the equilibrium strategies are as follows.
Each bidder $i$ will bid $(\bar{v}_i,0)$ and set $t_e=t_i$, i.e., bidder $i$ will set her fee cap equal to her valuation, bid a negligible amount, and set the time limit to one block.
The producer will include all eligible transactions immediately.
\end{proposition}
\begin{proof}
From equation \ref{eq:bidder_utility} it is clear that a bidder (weakly) prefers immediate inclusion to delayed inclusion, even for high values of $\Delta_i$.
On the other hand, if transaction $i$ sets $t_e=t_i$, then it is a best response for the producer to include $i$ immediately and collect the tip $\tau_i$, as it will not be possible to include it in later blocks.
This implies that a bidder's dominant strategy in this case is to always set $t_e=t_i$.
The fact that $b^c_i=\bar{v}_i$, a tip $\tau_i=0$ follows from \cite{tfmd}, as this is not affected by the introduction of the time limit.
\end{proof}

By introducing time-bound signatures, we have recouped the predicted steady-state equilibrium for patient transactions.


\section{Conclusion}\label{sec:conclusions}
We have introduced a simple yet effective modification to the Schnorr signature scheme that enforces a block‐height–based expiry, leveraging the blockchain itself as a tamper‐resistant clock. Our time‐bound signatures retain the same efficiency and security guarantees as vanilla Schnorr under the discrete logarithm assumption, while sharply reducing non‐myopic MEV opportunities by forcing immediate inclusion or forfeiture. Through security proofs in the AGM, and a Stackelberg game analysis of EIP‐1559, we demonstrated that time‐bound signatures recover the predicted steady‐state equilibrium of negligible tips. This work opens the door to practical, low‐overhead MEV countermeasures deployable on existing blockchains, and suggests further exploration of time‐bound primitives in other auction and consensus settings. Said signature scheme can be trivially applied to other blockchains and TFMs as well as adapted to allow for smart contract timeouts and escrow protocols.
\bibliographystyle{IEEEtran}
\bibliography{bib}
\end{document}